\documentclass[11pt]{article}
\usepackage{amsmath, amssymb, amscd, amsthm, amsfonts}
\usepackage{graphicx}
\usepackage{hyperref}

\title{On the closed-form expected NPVs of double barrier strategies for regular diffusions}
\author{Chongrui Zhu \thanks{Corresponding author: chongruizhu5@gmail.com}}
\date{}

\newtheorem{theorem}{Theorem}

\newtheorem{lemma}[theorem]{Lemma}
\newtheorem{remark}{Remark}
\newtheorem{proposition}{Proposition}

\newtheorem{assumption}[theorem]{Assumption}

\begin{document}

\maketitle

\begin{abstract}
The core of the research is to provide the explicit expression for the expected net present values (NPVs) of double barrier strategies for regular diffusions on the real line without solving differential equations. Under the so-called bail-out setting, the value of the expected NPVs of an insurance company varies according to the choice of a pair of policies, which consist of dividend payments paid out and capital injections received. In the case of the double barrier strategy, the expected NPVs are expressible with the help of certain types of functions allowing explicit expression in some cases, which is called the bivariate $q$-scale function in the article. This is accomplished by making use of a perturbation technique in \cite{czarna2014dividend}, which could lead to the linear equation system. In addition, a condition ensuring the existence of an optimal (upper) barrier level is presented. In the end, examples fitting the condition for selecting the optimal barrier are given.  
\end{abstract}

\section{Introduction}

In recent years, of much significance to the dividends control community are applications of fluctuation theory for spectrally negative L\'evy processes (SNLPs). To be more precise, the so-called $q$-scale function, as the eigen-function of SNLPs, plays an increasingly crucial part in giving the analytical form of the expected NPVs in various problem settings. A range of the representative literature on this topic can be found in \cite{loeffen2010finetti,bayraktar2014optimal,avanzi2016optimal,noba2018optimal} and references therein. The question arising out of the studies aforesaid naturally is whether the equivalent of the $q$-scale function exists regarding differing classes of stochastic processes and, additionally, whether the idea of utilizing that to give the closed-form expected NPVs is plausible or not. In response to the question, the case of the double barrier strategy for regular diffusions is under consideration in this paper.

The risk surplus process is the regular diffusion whose drift and volatility parameters depend on the surplus level itself. The so-called bail-out setting initiated in \cite{avram2007optimal}, which is the offshoot of classical de Finetti setting, modifies the controlled risk reserve with the capital perpetually injected so that it could never hit below zero level. Under this setting, over an infinite time horizon the cost term concerning the cumulative discounted volume of capital injected is incurred in our expected NPVs which contain the cumulative discounted dividend payments.

The treatment of two-sided exit problem for regular diffusions in \cite{lehoczky1977formulas,zhang2015occupation} offers an enlightenment for suggesting the analytical expression for the expected NPVs. A specific bivariate function represented by the two fundamental solutions to the Sturm-Liouvile equation characterizes the solution to the two-sided exit problem. The double barrier strategy concerning both a positive barrier at $a>0$ and the zero barrier would bring the surplus being out of the interval $[0,a]$ back to the adjacent barrier level. Making use of the inextricable connection between the double barrier strategy and two-sided exit problem, the expected NPVs are expressed explicitly with the help of the bivariate function. Here, for simplicity, we name this function by \textit{bivariate $q$-scale functions for regular diffusions}, which could be thought of as the analogue of the $q$-scale function for SNLPs in our case.

The realization of the above-mentioned expression seems to be unwarranted for regular diffusions not least because of the lack of the spatial homogeneity property, which, as well as the strong Markov property, proves to be the necessary tool for establishing  explicit expected NPVs based upon the research by \cite{avram2004exit,BIFFIS201085} and is expounded as one of the fundamental characteristics for general L\'evy processes in \cite{ken1999levy}. To specify, for a given scale-valued stochastic process, the spatial homogeneity states that  $Z=\{Z_t,t\geq 0\}$, we have
\begin{gather*}
    \{Z_t,t\geq 0;Z_0=x\}
    \text{  is equal in law to  }
    \{Z_t+x,t \geq 0;Z_0=0\},
\end{gather*}
for $x\in \mathbb{R}$. The exact "fine" properties as to the trajectory of regular diffusions that we could possibly rely on in obtaining closed-form manifestation are just the strong Markov property and continuity of sample path. Nevertheless, in view of the peculiar construction of the double barrier strategy in conjunction with the aforementioned property, the values of the expected NPVs at two barrier levels actually fulfill the linear equations innately, which almost yields the desired. It is also noteworthy that there is no need to involve excursion theory in this paper.

Also, the survey includes an optimization result, which says that, for certain types of diffusions, the expected NPVs in question could be optimized by selecting a suitable (upper) barrier level. Thanks to the closed-form expression, the candidate barrier level is chosen to be the unique zero point of the certain function represented in terms of bivariate $q$-scale functions. 

The rest of the paper is structured as follows: in Section \ref{sec2} we introduce the bivariate $q$-scale function for regular diffusions and give the mathematical definition of the expected NPVs and the construction of the double barrier strategy. In Section \ref{sec3}, we present our main results of the closed-form expected NPVs of double barrier strategies. In Section \ref{SC4}, we provide examples of Ornstein-Uhlenbeck processes and diffusions with the exponentially decayed mean-reversion drift as they fit the condition of the optimization result.

\section{Preliminaries}\label{sec2}
\subsection{Bivariate \textit{q}-scale functions for regular diffusions}
Let the risk surplus process $X=\{X_t:t\geq 0\}$ be a regular diffusion defined on the filtered probability space $(\Omega,\mathcal{F},\mathbb{F}=(\mathcal{F}_t)_{t\geq 0},\mathbb{P})$ satisfying the common assumption. Mathematically, $X$ is given by 
\begin{gather*}
X_t=x+
\int_0^t\mu(X_s)ds
+\int_0^t\sigma(X_s)dB_s
,\quad x\in \mathbb{R},\quad t\geq 0,
\tag{2.1}
\label{2.1}
\end{gather*}
in which $\mu$ and $\sigma$ are continuous real-valued functions on $(-\infty,+\infty)$ that satisfy the common condition ensuring the existence and uniqueness of \eqref{2.1}, and $B$ is the standard one-dimensional Brownian motion. In addition, the following restrictions are enforced to the process $X$:

\begin{assumption}[Non-degeneracy]
\label{A.1}
The parameter $\sigma:(-\infty,\infty)\to (-\infty,+\infty)$ is such that 
    $$\sigma>0 \text{  on  }(-\infty,+\infty).$$
\end{assumption}

\begin{assumption}[Local-integrability]
\label{A.2}
The drift parameter $\mu:(-\infty,+\infty)\to (-\infty,+\infty)$ and the volatility paramter $\sigma:(-\infty,\infty)\to (-\infty,+\infty)$ are satisfying that  
    $$\int_{x-\varepsilon}^{x+\varepsilon}\frac{1+|\mu(s)|}{\sigma^2(s)}ds<+\infty,\quad \text{for all  }x-\varepsilon,x,x+\varepsilon\in(-\infty,+\infty) \text{  with some  }\varepsilon>0.$$
\end{assumption}

The probability law for the process $X$ issued from $x$ is written as $\mathbb{P}_x$, and the corresponding mathematical expectation is denoted by $\mathbb{E}_x$. The infinitesimal generator of the process $X$ is the operator $\mathcal{G}$ on $C^2(\mathbb{R})$ with
\begin{gather*}
    \mathcal{G} f(x)
    =\frac{1}{2}\sigma^2(x)f^{\prime\prime}(x)+\mu(x)f^{\prime}(x),
    \quad 
    x\in \mathbb{R},
\end{gather*}
and the associated Sturm-Liouvile equation for $f\in C^2(\mathbb{R})$ is given by
\begin{gather*}
\left(\mathcal{G}-q\right)f(x)=0,\quad x\in \mathbb{R},
\tag{2.2}
\label{2.2}
\end{gather*}
whose solution would be called the eigen-function of $X$. The equation \eqref{2.2} admits two fundamental positive solutions: $\phi_q^+$ (strictly increasing) and $\phi_q^-$ (strictly decreasing) if $q>0$. Moreover, the functions $\phi_q^+$ and $\phi_q^-$ can be fixed as 
\begin{gather*}
    \phi_q^+(x)
    =
    \left\{
    \begin{aligned}&
    \mathbb{E}_x\left[e^{-q\tau_p}\right]
    ,&\quad x\in (-\infty,p]\\
    &
    \left(\mathbb{E}_p\left[e^{-q\tau_x}\right]\right)^{-1}
    ,&\quad x\in (p,\infty)
    \end{aligned}
    \right.
    ,\quad 
    \phi_q^-(x)
    =
    \left\{
    \begin{aligned}&
    \left(\mathbb{E}_p\left[e^{-q\tau_x}\right]\right)^{-1},&\quad x\in (-\infty,p]\\
    &
     \mathbb{E}_x\left[e^{-q\tau_p}\right]
    ,&\quad x\in (p,\infty)
    \end{aligned},
    \right.
\end{gather*}
for arbitrarily given $p\in \mathbb{R}$, where the first hitting times of $X$ is taken as
\begin{gather*}
    \tau_x:=\inf\{t>0:X_t=x\},\quad x\in \mathbb{R},
\end{gather*}
with the convention that $\inf\emptyset=\infty$. The function $s$ named as the scale function for regular diffusion $X$ satisfies \eqref{2.1} when $q=0$. Furthermore, its derivative $s^{\prime}$ can be written as 
\begin{gather*}
    s^{\prime}(x)=e^{-\int^x_{k}
    \frac{2\mu(s)}{\sigma^2(s)}ds}
    =\frac{(\phi^+_q)^{\prime}(x)\phi_q^-(x)-(\phi_q^-)^{\prime}(x)\phi^+_q(x)}{c_q},
    \tag{2.3}
\label{2.3}
\end{gather*}
for some $k\in (-\infty,+\infty)$ and the constant $c_q>0$ which is independent of $x$.
Subsequently, define the function $W_{(q)}:\mathbb{R}^2\to \mathbb{R}$ by
\begin{gather*}
    W_{(q)}(x,y)
    :=\frac{\phi_q^+(x)\phi_q^-(y)-\phi_q^-(x)\phi^+_q(y)}{c_q}
    ,\quad \text{ for } q>0.
    \tag{2.4}
    \label{2.4}
\end{gather*}
It is worth mentioning that the following easy-to-check property of function $W_{(q)}$ in Lemma \ref{L1} would be used throughout the paper without being referred to specifically. 
\begin{lemma}\label{L1}
Given that the functions $\phi^+_{q}$ and $\phi^-_{q}$ are sufficiently smooth on $(-\infty,+\infty)$. For all $x,y,u,z\in \mathbb{R}$,
we have 
\begin{gather*}
    W_{(q)}(x,y)=-W_{(q)}(y,x),\quad 
    \frac{\partial^2 W_{(q)}(u,z)}{\partial u \partial z}
    =-\frac{\partial^2 W_{(q)}(z,u)}{\partial z \partial u},\quad
    \frac{\partial^3 W_{(q)}(u,z)}{\partial u^2 \partial z}
    =-\frac{\partial^3 W_{(q)}(z,u)}{\partial z \partial u^2}
    .
\end{gather*}
Especially, $W_{(q)}(x,x)=\frac{\partial^2 W_{(q)}(u,z)}{\partial u \partial z}|_{z=u=x}=0$ for all $x\in \mathbb{R}$.
\end{lemma}
The result regarding two-sided exit problem in \cite{lehoczky1977formulas} is restated here:
\begin{lemma}
If $(y-x)(y-z)<0$ and $q\geq 0$, it holds that 
\begin{gather*}
    \mathbb{E}_y[e^{-q\tau_z}1_{\{\tau_z<\tau_x\}}]
    =\frac{W_{(q)}(y,x)}{W_{(q)}(z,x)}
    =\frac{W_{(q)}(x,y)}{W_{(q)}(x,z)}
    .
    \tag{2.5}
    \label{2.5}
\end{gather*}
\end{lemma}
\begin{remark}
An extensive set of specific examples of two fundamental solutions to \eqref{2.2}, i.e. $\phi_q^+$ and $\phi_q^{-}$, is available in \cite{borodin2015handbook}. In those examples, the corresponding scale function $s$ and constant $c_q$ are also given. The way of defining the function $W_{(q)}$ here could also be found in \cite{zhang2015occupation}.
\end{remark}
In the sequel, we shall make the convention that 
\begin{gather*}
W_{(q)}^1(u,z)=
\frac{\partial W_{(q)}(u,z)}{\partial u}
=
\frac{(\phi^+_q)^{\prime}(u)\phi_q^-(z)-(\phi_q^-)^{\prime}(u)\phi^+_q(z)}{c_q},
\\
W_{(q)}^{12}(u,z)=
\frac{\partial^2 W_{(q)}(u,z)}{\partial u\partial z}
=
\frac{(\phi^+_q)^{\prime}(u)(\phi_q^-)^{\prime}(z)-(\phi_q^-)^{\prime}(u)(\phi^+_q)^{\prime}(z)}{c_q},
\\
W_{(q)}^{122}(u,z)=
\frac{\partial^3 W_{(q)}(u,z)}{\partial u\partial z^2}
=
\frac{(\phi^+_q)^{\prime}(u)(\phi_q^-)^{\prime\prime}(z)-(\phi_q^-)^{\prime}(u)(\phi^+_q)^{\prime\prime}(z)}{c_q},
\\
W_{(q)}^{112}(u,z)=
\frac{\partial^3 W_{(q)}(u,z)}{\partial u^2 \partial z}
=
\frac{(\phi^+_q)^{\prime\prime}(u)(\phi_q^-)^{\prime}(z)-(\phi_q^-)^{\prime\prime}(u)(\phi^+_q)^{\prime}(z)}{c_q},
\\
    \mathbb{E}_{y}
    \left[e^{-\tau_z}1_{\{\tau_z<\tau_x\}}\right]
    =\overline{\psi}_{x,z}(y),\quad 
    \mathbb{E}_{y}
    \left[e^{-\tau_x}1_{\{\tau_x<\tau_z\}}\right]
    =\underline{\psi}_{x,z}(y),
\end{gather*}
for all $u,z\in \mathbb{R}$ and $x<y<z$.
\subsection{The definition of the expected NPVs under the bail-out setting}
The process $U^\pi$ controlled by the policy pair $\pi=\{(D_t^\pi,R_t^\pi):t\geq 0\}$, which consists of dividend payments $D^{\pi}=\{D_t^{\pi}:t\geq 0\}$ and capital injections $R^{\pi}=\{R_t^{\pi}:t\geq 0\}$, is formulated as
\begin{align*}
    U_t^{\pi}=X_t-D_t^{\pi}+R_t^{\pi},\quad t\geq 0,
\end{align*}
where $D^\pi$ and $R^\pi$ is non-decreasing and $\mathbb{F}-$adapted. It is also to be noted that here both the cumulative dividend payments $D^{\pi}$ and the volume of capital injected $R^{\pi}$ are right-continuous processes, starting from $0$. To illustrate in more detail, $\Pi$ is the admissible class that consists of the dividend policies $\pi$ such that 
\begin{align*}
    U_t^{\pi}\geq 0,\quad \text{for all}\quad t\geq 0, \quad \text{and} \quad  V_R^{\pi}(x)<\infty,\quad \text{a.s.}
\end{align*}
In the control literature, the aim is to identify the strategy $\pi\in \Pi$ that is able to maximize
\begin{align}
    V^{\pi}(x)=V_{D}^{\pi}(x)-\varphi V_R^{\pi}(x)
    ,\quad x\geq 0,
    \tag{2.6}
    \label{2.6}
\end{align}
where $\varphi>1$ is the unit cost for the capital injection, $V_D^{\pi}$ and $V_R^{\pi}$ are formulated as
\begin{align}
    V_{D}^{\pi}(x)
    =
    \mathbb{E}_{x}
    \left[
    \int_0^{\infty} 
    e^{-qt}dD_t^\pi
    \right],\quad 
     V_R^{\pi}(x)
    =
    \mathbb{E}_{x}
    \left[
    \int_0^{\infty} 
    e^{-qt}dR_t^\pi
    \right]
    ,\quad x\geq 0,
    \tag{2.7}
    \label{2.7}
\end{align}
in which $q>0$ is the discounting factor. Nevertheless, here we restrict ourselves to computing the value of \eqref{2.6} when the admissible policy is the double barrier strategy.
\subsection{Construction of the double barrier strategy}
The exact formulation of the double barrier strategy $\pi_a=(D^a,R^a)$ is given as follows:
\begin{enumerate}
    \item To begin with, let $\sigma_0=\tau_0$, $\sigma_a=\tau_a$ and $D_t^a=R_t^a=0$ for $t<\sigma_0\wedge \sigma_a$. If $X_{\sigma_0\wedge \sigma_a}=0$, then go to step 2. Otherwise, go to step 3.
    \item Define
    \begin{gather*}
        \overline{R}_t=
        -\inf_{\sigma_0\leq s \leq t}(X_s\wedge 0),\quad 
        \overline{U}_t=X_t-\inf_{\sigma_0\leq s \leq t}(X_s\wedge 0),
    \end{gather*}
    for $t\geq \sigma_0$. Set $\sigma_a=\inf\{t>0:\overline{U}_t=a\}$. If $\sigma_0 \leq t<\sigma_a$, let $D_t^a=D_{\sigma_0-}^a$, $R_t^a=R_{\sigma_0-}^a+\overline{R}_t$ and $U_t^a=\overline{U}_t$. Then go to step 3.
    \item Define
    \begin{gather*}
        \overline{D}_t
        =\sup_{\sigma_a\leq s \leq t}[(X_s-a)\vee 0],
        \quad 
        \overline{U}_t
        =X_t-\sup_{\sigma_a \leq s \leq t}
        \left[(X_s-a)\vee 0\right],
    \end{gather*}
    for $t\geq \sigma_a$. Set $\sigma_0=\inf\{t>0:\overline{U}_t=0\}$. If $\sigma_a\leq t<\sigma_0$, let $D_t^a=D_{\sigma_a-}^a+\overline{D}_t$, $R_t^a=R_{\sigma_a-}^a$ and $U_t^a=\overline{U}_t$. Then go to step 2.
\end{enumerate}
The strategy $\pi_a$ is also such that the support of measure $dD_t^a$ and $dR_t^a$ is included in the set of $\overline{\{U_t^a=a\}}$ and $\overline{\{U_t^a=0\}}$, respectively. The definition of the strategy $\pi_a$ here is rephrased from \cite{avram2007optimal}. Next, $V^{\pi_a}$, the expected NPVs function, would be abbreviated to $V^a$.
\section{Main results: the closed-form expected NPVs}\label{sec3}
\subsection{Representation of the expected NPVs}
In order to give the explicit expression of $V^a$ in Theorem \ref{T.1}, we prove an associated specific result of the function $W_{(q)}$ here firstly.
\begin{lemma}\label{L.3}
For all $a>0$, we have 
\begin{gather*}
    W_{(q)}^{12}(0,a)>0.
    \tag{3.1}
    \label{3.1}
\end{gather*}
\end{lemma}
\begin{proof}
Since the functions $\phi_{q}^+$ and $\phi_{q}^-$ are the eigen-functions of $X$, i.e. $\phi_q^+$ and $\phi_q^-$ satisfy \eqref{2.2}, for $x\geq 0$ we have 
\begin{gather*}
    (\phi_q^{+})^{\prime\prime}(x)
    +
    \frac{2\mu(x)}{\sigma^2(x)}(\phi_q^+)^{\prime}(x)
    =\frac{2q}{\sigma^2(x)}\phi_q^{+}(x),\\
    (\phi_q^{-})^{\prime\prime}(x)
    +
    \frac{2\mu(x)}{\sigma^2(x)}(\phi_q^-)^{\prime}(x)
    =\frac{2q}{\sigma^2(x)}\phi_q^{-}(x),
\end{gather*}
as $\sigma$ is strictly positive on $[0,\infty)$. Easily, we deduce that 
\begin{gather*}
    \left[(\phi_q^+)^{\prime}(t)
    e^{\int_0^t\frac{2\mu(s)}{\sigma^2(s)}ds}
    \right]^{\prime}
    \bigg|_{t=x}=\frac{2q}{\sigma^2(x)}e^{\int_0^x\frac{2\mu(s)}{\sigma^2(s)}ds}\phi_q^{+}(x),
    \tag{3.2}
    \label{3.2}
    \\
    \left[(\phi_q^-)^{\prime}(t)
    e^{\int_0^t\frac{2\mu(s)}{\sigma^2(s)}ds}
    \right]^{\prime}
    \bigg|_{t=x}=\frac{2q}{\sigma^2(x)}e^{\int_0^x\frac{2\mu(s)}{\sigma^2(s)}ds}\phi_q^{-}(x).
    \tag{3.3}
    \label{3.3}
\end{gather*}
Integrating both sides of \eqref{3.2} and \eqref{3.3} on $[0,a]$ respectively gives 
\begin{gather*}
(\phi_q^+)^{\prime}(a)
    e^{\int_0^a\frac{2\mu(s)}{\sigma^2(s)}ds}=(\phi_q^+)^{\prime}(0)
    +\int_0^a \frac{2q}{\sigma^2(x)}e^{\int_0^x\frac{2\mu(s)}{\sigma^2(s)}ds}\phi_q^{+}(x)dx,
    \tag{3.4}
    \label{3.4}
    \\
    (\phi_q^-)^{\prime}(a)
    e^{\int_0^a\frac{2\mu(s)}{\sigma^2(s)}ds}=(\phi_q^-)^{\prime}(0)
    +\int_0^a \frac{2q}{\sigma^2(x)}e^{\int_0^x\frac{2\mu(s)}{\sigma^2(s)}ds}\phi_q^{-}(x)dx,
    \tag{3.5}
    \label{3.5}
\end{gather*}
where the right-hand side of \eqref{3.4} and \eqref{3.5} is well-defined due to Assumption \ref{A.1} and Assumption \ref{A.2}. Thereby based upon \eqref{3.4} and \eqref{3.5} we have 
\begin{gather*}
e^{\int_0^a\frac{2\mu(s)}{\sigma^2(s)}ds}W_{(q)}^{12}(0,a)
=e^{\int_0^a\frac{2\mu(s)}{\sigma^2(s)}ds}
\frac{
(\phi_q^+)^{\prime}(0)(\phi^-)^{\prime}(a)
-
(\phi_q^-)^{\prime}(0)(\phi_q^{+})^{\prime}(a)
}{c_q}
\\
=\int_0^a \frac{2q}{\sigma^2(x)}e^{\int_0^x\frac{2\mu(s)}{\sigma^2(s)}ds}W_{(q)}^1(0,x)dx>0,
\end{gather*}
where $$W_{(q)}^1(0,x)=\frac{(\phi_q^{+})^{\prime}(0)\phi_q^-(x)
    -(\phi_q^{-})^{\prime}(0)\phi_q^+(x)}{c_q}>0,$$ for all $x\in \mathbb{R}$, since functions $\phi_{q}^+$ and $\phi_q^-$ are both positive, $(\phi_{q}^+)^{\prime}>0$ and $(\phi_{q}^-)^{\prime}<0$. Thus we finalize the proof.
\end{proof}
\begin{theorem}\label{T.1}
The value function defined in \eqref{2.6} for the double barrier strategy at positive level $a> 0$ is such that
\begin{gather*}
    V^a(x)=
    \left\{
    \begin{aligned}&
    V^a(0)+\varphi x,&\quad x\in (-\infty,0),\\
    &
   \frac{
    W_{(q)}^1(0,x)
    -\varphi W_{(q)}^1(a,x)
    }{W_{(q)}^{12}(0,a)}
    ,&\quad x\in [0,a],
    \\
    &
    V^a(a)+x-a
    ,&\quad x\in (a,\infty).
    \end{aligned}
    \tag{3.6}
    \label{3.6}
    \right.
\end{gather*}
\end{theorem}
\begin{proof}
Usually, the establishment of identities that are similar to \eqref{3.6} heavily relies on the excursion theory. Nevertheless, we show that it can actually be eschewed on account of the perturbation technique in \cite{czarna2014dividend}. 

Let $\varepsilon>0$. In view of the construction of the double barrier strategy and strong Markov property of $X$ at $\tau_{-\varepsilon}\wedge\tau_{a}$ and $\tau_0\wedge \tau_{a+\varepsilon}$, correspondingly, we could deduce that 
\begin{align*}
    V^a(0)&=-\varphi\mathbb{E}_0\left[
    \int_0^{{ \tau_{-\varepsilon}\wedge \tau_{a}}} e^{-qt}
    dR_t^a
    \right]+
    \mathbb{E}_0\left[
    \int_{{ \tau_{-\varepsilon}\wedge \tau_{a}}}^\infty e^{-qt}
    dD_t^a
    \right]-\varphi
    \mathbb{E}_0\left[
    \int_{{ \tau_{-\varepsilon}\wedge \tau_{a}}}^\infty e^{-qt}
    dR_t^a
    \right]
    \\
    &=-\varphi\mathbb{E}_0\left[
    \int_0^{ \tau_{-\varepsilon}\wedge \tau_{a}} e^{-qt}
    dR_t^a
    \right]+
    \overline{\psi}_{a,-\varepsilon}(0)
    V^a(a)
    +
    \underline{\psi}_{a,-\varepsilon}(0)
    V^a(-\varepsilon)
    \\
    &=-\varphi\mathbb{E}_0\left[
    \int_0^{ \tau_{-\varepsilon}\wedge \tau_{a}} e^{-qt}
    dR_t^a
    \right]+
    \overline{\psi}_{a,-\varepsilon}(0)
    V^a(a)
    +
    \underline{\psi}_{a,-\varepsilon}(0)
    \left[V^a(0)-\varphi\varepsilon\right],
\end{align*}
and
\begin{align*}
    V^a(a)&
    =\mathbb{E}_a\left[
    \int_0^{\tau_0\wedge \tau_{a+\varepsilon}} e^{-qt}
    dD_t^a
    \right]+
    \mathbb{E}_a\left[
    \int_{ \tau_0\wedge \tau_{a+\varepsilon}}^\infty e^{-qt}
    dD_t^a
    \right]
    -\varphi
    \mathbb{E}_a\left[
    \int_{ \tau_0\wedge \tau_{a+\varepsilon}}^\infty e^{-qt}
    dR_t^a
    \right]
    \\
    &=\mathbb{E}_a\left[
    \int_0^{ \tau_0\wedge \tau_{a+\varepsilon}} e^{-qt}
    dD_t^a
    \right]+
    \overline{\psi}_{a+\varepsilon,0}(a)
    V^a(a+\varepsilon)
    +
    \underline{\psi}_{a+\varepsilon,0}(a)
    V^a(0)
    \\
&=\mathbb{E}_a\left[
    \int_0^{\tau_0\wedge \tau_{a+\varepsilon}} e^{-qt}
    dD_t^a
    \right]+
    \overline{\psi}_{a+\varepsilon,0}(a)
    \left[V^a(a)+\varepsilon
    \right]
    +
    \underline{\psi}_{a+\varepsilon,0}(a)V^a(0).
\end{align*}
Furthermore, we claim that $\mathbb{E}_0\left[
    \int_0^{{ \tau_{-\varepsilon}\wedge \tau_{a}}} e^{-qt}
    dR_t^a
    \right]=o(\varepsilon)$ and $\mathbb{E}_a\left[
    \int_0^{\tau_0\wedge \tau_{a+\varepsilon}} e^{-qt}
    dD_t^a
    \right]=o(\varepsilon)$ by
\begin{gather*}
0\leq \mathbb{E}_0\left[
    \int_0^{ \tau_{-\varepsilon}\wedge \tau_{a}} e^{-qt}
    dR_t^a
    \right]
    \leq 
    \varepsilon
    \mathbb{E}_0\left[
    \int_0^{ \tau_{-\varepsilon}\wedge \tau_{a}} e^{-qt}
    dt
    \right]
    \leq 
    \varepsilon
    \mathbb{E}_0\left[
    \int_0^{ \tau_{-\varepsilon}} e^{-qt}
    dt
    \right]
    \\
    \leq 
    \frac{\varepsilon}{q}
    \left[
    1-
    \mathbb{E}_0\left[
    e^{-q\tau_{-\varepsilon}}
    1_{\{\tau_{-\varepsilon}<\infty\}}
    \right]
    \right];
    \quad 
    \lim_{\varepsilon\to 0^+}\mathbb{E}_0\left[
    e^{-q\tau_{-\varepsilon}}
    1_{\{\tau_{-\varepsilon}<\infty\}}
    \right]=1,
    \tag{3.7}
    \label{3.7}
\end{gather*}
and 
\begin{gather*}
0\leq \mathbb{E}_a\left[
    \int_0^{\tau_0\wedge \tau_{a+\varepsilon}} e^{-qt}
    dD_t^a
    \right]\leq 
    \varepsilon
    \mathbb{E}_a\left[
    \int_0^{\tau_0\wedge \tau_{a+\varepsilon}} e^{-qt}
    dt
    \right]
    \leq 
    \varepsilon\mathbb{E}_a\left[
    \int_0^{\tau_{a+\varepsilon}} e^{-qt}
    dt
    \right]
    \\
    \leq 
    \frac{\varepsilon}{q}
    \left[
    1-
    \mathbb{E}_a
    \left[
    e^{-q\tau_{a+\varepsilon}}
    1_{\{\tau_{a+\varepsilon}<\infty\}}
    \right]
    \right];
    \quad
    \lim_{\varepsilon\to 0^+}\mathbb{E}_a
    \left[
    e^{-q\tau_{a+\varepsilon}}
    1_{\{\tau_{a+\varepsilon}<\infty\}}
    \right]=1,
    \\
    \tag{3.8}
    \label{3.8}
\end{gather*}
since the increment of $R^a$ under $\mathbb{P}_0$ would not be over $\varepsilon$ before the epoch $\tau_{-\varepsilon}\wedge \tau_a$ and $D_a$ under $\mathbb{P}_a$ could only increase up to $\varepsilon$ before the moment $\tau_0\wedge \tau_{a+\varepsilon}$. Dividing both sides of \eqref{3.7} and \eqref{3.8} by $\varepsilon>0$ and letting $\varepsilon\to0^+$ respectively yields the linear equation system that 
\begin{gather*}
    V^a(0)=
    \frac{\varphi W_{(q)}(0,a)+W_{(q)}^1(0,0)V^a(a)}
    {W_{(q)}^1(0,a)},
    \quad
    V^a(a)=
    \frac{W_{(q)}(a,0)+W_{(q)}^1(a,a)V^a(0)}
    {W_{(q)}^1(a,0)}
    .
\end{gather*}
After careful calculations involving the identity given by
\begin{gather*}
W_{(q)}(a,0)W_{(q)}^{12}(0,a)+W_{(q)}^1(a,a)W_{(q)}^1(0,0)=W_{(q)}^1(0,a)W_{(q)}^1(a,0),
\end{gather*}
which is shown by invoking the definition of the function $W_{(q)}$, we have 
\begin{gather*}
    V^a(0)
    =
    \frac{W_{(q)}^1(0,0)-\varphi W_{(q)}^1(a,0)}
    {W_{(q)}^{12}(0,a)}
    ,\quad 
    V^a(a)=\frac{W_{(q)}^1(0,a)-\varphi W_{(q)}^1(a,a)}
    {W_{(q)}^{12}(0,a)},
\end{gather*}
where $W_{(q)}^{12}(0,a)$ is strictly positive owing to Lemma \ref{L.3}. By the construction of the double barrier strategy $\pi_a$ and the strong Markov property of $X$ at $\tau_0\wedge \tau_a$, for $x\in(0,a)$ we infer that
\begin{gather*}
    V^a(x)
    =\mathbb{E}_x
    \left[
    \int_0^\infty e^{-qt}
    dD_t^a
    \right]-
    \varphi
    \mathbb{E}_x
    \left[
    \int_0^\infty e^{-qt}
    dR_t^a
    \right]
    =
    \mathbb{E}_x
    \left[
    \int_{{\tau}_{0}\wedge {\tau}_{a}}^\infty e^{-qt}
    dD_t^a
    \right] -
    \varphi
    \mathbb{E}_x
    \left[
    \int_{{\tau}_{0}\wedge {\tau}_{a}}^\infty e^{-qt}
    dR_t^a
    \right]
    \\
    =
    \overline{\psi}_{a,0}(x)V^a(a)
    +\underline{\psi}_{a,0}(x)V^a(0)
    \\
    =\frac{
    \left[ 
    W_{(q)}(x,0)W_{(q)}^1(0,a)
    +W_{(q)}(a,x)W_{(q)}^1(0,0)
    \right]
    -\varphi
    \left[
    W_{(q)}(x,0)W_{(q)}^1(a,a)
    +W_{(q)}(a,x)W_{(q)}^1(a,0)
    \right]
    }{
    W_{(q)}(a,0)W_{(q)}^{12}(0,a)
    }
    \\=
    \frac{
    W_{(q)}(a,0)W_{(q)}^1(0,x)
    -\varphi W_{(q)}(a,0)W_{(q)}^1(a,
    x)
    }{W_{(q)}(a,0)W_{(q)}^{12}(0,a)}
    \\
    =
    \frac{
    W_{(q)}^1(0,x)
    -\varphi W_{(q)}^1(a,x)
    }{W_{(q)}^{12}(0,a)}
    ,
\end{gather*}
where the penultimate equality is due to the calculation based upon the definition of $W_{(q)}$ in \eqref{2.4}. The value of function $V^a$ on $(-\infty,0)\cap(a,\infty)$ is obtained by the construction of strategy $\pi_a$.
\end{proof}
\begin{remark}
Notably, Theorem \ref{T.1} is not valid for the diffusion with the state space $[0,\infty)$ or $(0,\infty)$, such as Bessel processes and Geometric Brownian motions, in the sense that we have essentially made use of the hitting time $\tau_{-\varepsilon}$ with $\varepsilon>0$ within the proof of that theorem. De facto, with Lemma 2.1 given in \cite{zhang2015occupation} saying that $\lim_{x\to l}\phi_q^-(x)=\infty$, where $l$ is the left-end point of the state space, one can see that there is no general guarantee that $\lim_{\varepsilon\to 0+}W_{(q)}^{12}(\varepsilon,a)<\infty$ in such case for $l=0$ (Diffusions with the state space $[0,\infty)$ or $(0,\infty)$). Therefore, \eqref{3.6} could not be a generic formula for diffusions. 
\end{remark}
\begin{proposition}
It holds that 
\begin{gather*}
    (V^a)^{\prime}(a-)=1,\quad 
    (V^a)^{\prime}(0+)=\varphi.
\end{gather*}
\end{proposition}
\begin{proof}
Using the value of function $V^a$ in $(0,a)$ given in Theorem \ref{T.1} will lead to the result.
\end{proof}
\subsection{Selection of the (upper) barrier level}
An additional result is provided below in order to select the candidate barrier level $a^*>0$, maximizing \eqref{2.6} over a narrower set consisting of all double barrier strategies with the positive (upper) barrier level, which might also be instrumental in solving related control problems. Here, the way of selecting the level $a^*$ is inspired by Section 4 of \cite{NOBA202173}. Also, the rule of selecting $a^*$ only involves the function $W_{(q)}$.
\begin{lemma}
\label{lem5}
Define $V_x(a)=V^a(x)$ for all $a>0$ and $x\in\mathbb{R}$. We have 
\begin{gather*}
    V_x^{\prime}(a)
    =-
    W_{(q)}^1(0,x)
    \frac{
    W_{(q)}^{122}(0,a)+
    \varphi W_{(q)}^{112}(a,a)
    }{
    \left[W_{(q)}^{12}(0,a)\right]^2
    },
\end{gather*}
for $0\leq x\leq a$. If $x>a$, then $V_x^{\prime}(a)=V_a^{\prime}(a)$.
\end{lemma}
\begin{proof}
If $0\leq x\leq a$, differentiating \eqref{3.6} with respect to $a$ entails that 
\begin{gather*}
    V_x^{\prime}(a)
    =-\frac{W_{(q)}^1(0,x)
    W_{(q)}^{122}(0,a)
    +\varphi
    \left[
    W_{(q)}^{11}(a,x)W_{(q)}^{12}(0,a)
    -W_{(q)}^1(a,x)W_{(q)}^{122}(0,a)
    \right]
    }{
    \left[W_{(q)}^{12}(0,a)\right]^2
    }\\
    =-
    W_{(q)}^1(0,x)
    \frac{
    W_{(q)}^{122}(0,a)+
    \varphi W_{(q)}^{112}(a,a)
    }{
    \left[W_{(q)}^{12}(0,a)\right]^2
    },
\end{gather*}
where we have implemented the identity given by
\begin{gather*}
 W_{(q)}^{11}(c,d)W_{(q)}^{12}(b,a)
    -W_{(q)}^1(a,d)W_{(q)}^{122}(b,c)
    =W_{(q)}^1(b,d)W_{(q)}^{112}(c,a),
    \tag{3.9}
    \label{3.9}
\end{gather*}
when $b=0$, $c=a$ and $d=x$. Notice that \eqref{3.9} holds for all $b,c,d\in \mathbb{R}$, which could be checked by employing the definition of $W_{(q)}$. 

If $a<x$ instead, differentiating \eqref{3.6} with respect to $a$ yields that 
\begin{gather*}
    V_x^{\prime}(a)
    =\frac{
    \left[W_{(q)}^{12}(0,a)\right]^2
    -
    W_{(q)}^1(0,a)
    W_{(q)}^{122}(0,a)
    -\varphi
    \left[
    W_{(q)}^{11}(a,a)W_{(q)}^{12}(0,a)
    -W_{(q)}^1(a,a)W_{(q)}^{122}(0,a)
    \right]
    }{
    \left[W_{(q)}^{12}(0,a)\right]^2
    }-1
    \\
    =V_a^{\prime}(a),
\end{gather*}
where we have made use of the identity \eqref{3.9} for $b=0$, $c=a$ and $d=a$.
\end{proof}

Subsequently, define the candidate barrier level $a^*$ by 
\begin{gather*}
    a^*=\inf 
    \{
    a>0:
    \varsigma(a)=
    W_{(q)}^{122}(0,a)+
    \varphi W_{(q)}^{112}(a,a)>0
    \},
\end{gather*}
with the convention that $\inf\emptyset=+\infty$. The fact that selecting $a^*$ leads to the maximum value of the expected NPVs is based on the following proposition.

\begin{proposition}\label{P.2}
Suppose that the drift parameter $\mu$ and volatility parameter $\sigma$ are such that $\mu\in C^1(\mathbb{R})$ and $\sigma(\cdot)\equiv \sigma>0$. Additionally, assume that the functions $\phi_q^{+}$ and $\phi_q^{-}$ are at least three times continuously differentiable on $\mathbb{R}$. Presume that $\mu\leq 0$ and $\mu^{\prime}< q$ on $[0,\infty)$. Then $a^*\in(0,+\infty)$. Furthermore, $V^{a^*}(x)\geq V^a(x)$ for all $a>0$ and $x\in [0,\infty)$.
\end{proposition}
\begin{proof}
It follows from Lemma \ref{L.3} and its proof that 
\begin{gather*}
    W_{(q)}^{1}(x,y)>0,\quad\text{for all  } x,y\in \mathbb{R}, 
    \text{  and  }
    W_{(q)}^{12}(0,a)>0,\quad\text{for all  }a>0,
\end{gather*}
which shall be used in the sequel without specific justification.

To begin with, $W_{(q)}^{1222}(0,a)>0$ for all $a>0$ is the fact to be shown. In view of the proof of Lemma \ref{L.3}, we have
\begin{gather*}
    W_{(q)}^{12}(0,a)
=e^{-\int_0^a\frac{2\mu(s)}{\sigma^2}ds}\int_0^a \frac{2q}{\sigma^2}e^{\int_0^x\frac{2\mu(s)}{\sigma^2}ds}W_{(q)}^1(0,x)dx,
\tag{3.10}
\label{3.10}
\end{gather*}
for all $a>0$. Differentiating \eqref{3.10} leads to 
\begin{gather*}
    W_{(q)}^{122}(0,a)
    =\frac{2q}{\sigma^2}W_{(q)}^1(0,a)
    -\frac{2\mu(a)}{\sigma^2}W_{(q)}^{12}(0,a)>0,
    \tag{3.11}
\label{3.11}
\end{gather*}
since $\mu \leq 0$ on $[0,\infty)$. Afterwards, by the condition that $\mu \leq 0$ and $\mu^{\prime}< q$ on $[0,\infty)$, differentiating \eqref{3.11} would entail that 
\begin{gather*}
    W_{(q)}^{1222}(0,a)
    =\frac{2[q-\mu^{\prime}(a)]}{\sigma^2}W_{(q)}^{12}(0,a)
    -\frac{2\mu(a)}{\sigma^2}W_{(q)}^{122}(0,a)>0,
    \tag{3.12}
\label{3.12}
\end{gather*}
where the last inequality is obtained due to the condition mentioned and the fact that $W_{(q)}^{12}(0,a)>0$ and $W_{(q)}^{122}(0,a)>0$ for all $a>0$.

Next, we are to prove that $W_{(q)}^{1112}(a,a)>0$ for all $a>0$. Let $\varepsilon>0$. Integrating both sides of \eqref{3.2} and \eqref{3.3} on $[a-\varepsilon,a]$, we immediately deduce that
\begin{gather*}
(\phi_q^+)^{\prime}(a)
    e^{\int_0^a\frac{2\mu(s)}{\sigma^2}ds}
    -
(\phi_q^+)^{\prime}(a-\varepsilon)
    e^{\int_0^{a-\varepsilon}\frac{2\mu(s)}{\sigma^2}ds}
        =
    \int_{a-\varepsilon}^a \frac{2q}{\sigma^2}e^{\int_0^x\frac{2\mu(s)}{\sigma^2}ds}\phi_q^{+}(x)dx,
    \tag{3.13}
\label{3.13}
    \\
(\phi_q^-)^{\prime}(a)
    e^{\int_0^a\frac{2\mu(s)}{\sigma^2}ds}
    -
(\phi_q^-)^{\prime}(a-\varepsilon)
    e^{\int_0^{a-\varepsilon}\frac{2\mu(s)}{\sigma^2}ds}
        =
    \int_{a-\varepsilon}^a \frac{2q}{\sigma^2}e^{\int_0^x\frac{2\mu(s)}{\sigma^2}ds}\phi_q^{-}(x)dx.
    \tag{3.14}
\label{3.14}
\end{gather*}
In view of \eqref{3.13}, \eqref{3.14} and the fact that $W_{(q)}^1(x,a)> 0$ for $x\in[0,a]$ and that $\mu\leq 0$ on $[0,\infty)$, we obtain that
\begin{gather*}
e^{\int_0^a\frac{2\mu(s)}{\sigma^2}ds}
    \frac{W_{(q)}^1(a,a)-W_{(q)}^1(a-\varepsilon,a)}{\varepsilon}
    \leq
    \frac{W_{(q)}^1(a,a)
    e^{\int_0^a\frac{2\mu(s)}{\sigma^2}ds}
    -
W_{(q)}^{1}(a-\varepsilon,a)
    e^{\int_0^{a-\varepsilon}\frac{2\mu(s)}{\sigma^2}ds}}{\varepsilon}
    \\
    \leq 
    -\frac{
    \int_{a-\varepsilon}^a \frac{2q}{\sigma^2}e^{\int_0^x\frac{2\mu(s)}{\sigma^2}ds}
    W_{(q)}(a,x)dx
    }{\varepsilon}.
    \tag{3.15}
\label{3.15}
\end{gather*}
Letting $\varepsilon\to 0^+$ in \eqref{3.15}, we obtain $W_{(q)}^{11}(a,a)\leq 0$ because of the fact that $W_{(q)}(a,a)=0$ makes the right hand side vanish. Again by $W_{(q)}^{12}(a,a)=0$, \eqref{3.13}, and \eqref{3.14}, we have 
\begin{gather*}
W_{(q)}^{112}(a,a)=-\lim_{\varepsilon\to 0^+}
\frac{W_{(q)}^{12}(a-\varepsilon,a)}{\varepsilon}
\\
=-\lim_{\varepsilon\to 0^+}e^{-\int_0^{a-\varepsilon}\frac{2\mu(s)}{\sigma^2}ds}\frac{\int_{a-\varepsilon}^a
\frac{2q}{\sigma^2}e^{\int_0^x\frac{2\mu(s)}{\sigma^2}ds}
    W_{(q)}^1(a,x)dx
}{\varepsilon}=-\frac{2q}{\sigma^2}W_{(q)}^1(a,a)<0.
\tag{3.16}
\label{3.16}
\end{gather*}
Differentiating \eqref{3.16}, we have 
\begin{gather*}
    W_{(q)}^{1112}(a,a)
    =
    -\frac{2q}{\sigma^2}W_{(q)}^{11}(a,a)\geq 0,
    \tag{3.17}
    \label{3.17}
\end{gather*}
because of the fact that $W_{(q)}^{11}(a,a)\leq 0$ for all $a>0$. Subsequently, notice that 
\begin{gather*}
\varsigma(0)=
    W_{(q)}^{122}(0,0)+\varphi W_{(q)}^{112}(0,0)
    =(\varphi-1)W_{(q)}^{112}(0,0)=-(\varphi -1 )\frac{2q}{\sigma^2}W_{(q)}^1(0,0)<0.
\end{gather*}
Then $a^*> 0$. Let $\bar{a}>0$. Furthermore, from \eqref{3.12} and \eqref{3.17}, we have 
\begin{gather*}
\varsigma^{\prime}(a)=
W_{(q)}^{1222}(0,a)+\varphi W_{(q)}^{1112}(a,a)
    \geq W_{(q)}^{1222}(0,a)
    \\
    \geq 
    \frac{2 [q-\mu^{\prime}(a)] }{\sigma^2}W_{(q)}^{12}(0,a)>
    \frac{2 [q-\mu^{\prime}(a)]  }{\sigma^2}W_{(q)}^{12}(0,\bar{a})>0,
    \end{gather*}
for all $a>\bar{a}>0$, where the penultimate inequality is obtained by invoking that $W_{(q)}^{12}(0,\cdot)$ is strictly increasing on $(0,\infty)$ ($W_{(q)}^{122}(0,a)>0$ for all $a>0$), and the last inequality is achieved using \eqref{3.1}. Hence, $\varsigma(a)$ drifts to $+\infty$ as $a$ goes to infinity. Consequently, $0<a^*<+\infty$. Recalling Lemma \ref{lem5} shows that $V_x^{\prime}(a)$ has the sign as $\varsigma(a)$. As a result, for arbitrary $x\in[0,\infty)$, we obtain that $V_x^{\prime}(a)>0$ for $a\in (0,a^*)$ and $V_x^{\prime}(a)<0$ for $a\in (a^*,+\infty)$, which ends the proof.
\end{proof}
\begin{remark}
It is worth mentioning that in \cite{zhu2016optimal}, the double barrier strategy with the positive barrier level turns out to be non-optimal under the condition that $\sigma>0$ and $\mu^{\prime}<q$ on $(0,\infty)$, which in fact does not necessarily contradict with our result since \cite{zhu2016optimal} deals with an optimal control problem in a greater admissible class, which includes both strategies with and without capital injections. In our case, the goal is restricted to selecting the optimal one within the set of double barrier strategies. Furthermore, \cite{ferrari2019class} investigated the capital-injected dividend control problem with a proportional instantaneous reward. \cite{yin2012first} also derived the explicit value function of double barrier strategies for diffusions using It\^o's formula. 
\end{remark}

\section{Case study}
\label{SC4}
Both cases below ensure that Assumption \ref{A.1} and Assumption \ref{A.2} hold. Additionally, both two models are chosen so that the precondition of Proposition \ref{P.2} is satisfied, which means $a^*\in (0,+\infty)$ in both cases. 
\subsection{Ornstein-Uhlenbeck processes}
Let $\theta>0$. For the Ornstein-Uhlunbeck process described by the stochastic differential equation (SDE) as follows
\begin{gather*}
    X_t=x-\int_0^t \theta X_sds+B_t,\quad x\in \mathbb{R},
\end{gather*}
the two associated eigen-functions $\phi^+_q$ and $\phi^{-}_q$ could be chosen as
\begin{gather*}
    \phi_q^+(x)
    =e^{\frac{x^2\theta }{2}}
    D_{-\frac{q}{\theta}}(-\sqrt{2\theta}{x}),\quad 
    \phi_q^-(x)
    =e^{\frac{x^2\theta}{2}}
    D_{-\frac{q}{\theta}}(\sqrt{2\theta}{x}),\quad x\in \mathbb{R},
\end{gather*}
where $D$ is the parabolic cylinder function, whose property could be found in Appendix 1.22. and Appendix 2.6. of \cite{borodin2015handbook}. From Section 2.6 of \cite{bayraktar2010one} and the reference therein, the exact definition of $D$ is given by
\begin{gather*}
D=D_{v}(x)=2^{-\frac{v}{2}}e^{-\frac{x^2}{4}}H_v(\frac{x}{\sqrt{2}}),
\text{  for all  }x\in \mathbb{R},
\end{gather*}
where $H_v$ is the Hermite function defined as 
\begin{gather*}
    H_v(x)=\frac{e^{x^2}}{\Gamma(-v)}
    \int_0^{\infty}
    s^{-v-1}e^{-(s+x)^2}ds,\quad \text{Re}(v)<0.
\end{gather*}

The constant $c_q$ is given by $c_q=\frac{2\sqrt{\theta \pi}}{\Gamma(\frac{q}{\theta})}$ and the scale function is fixed as $s(x)=\int_0^x e^{\theta y^2}dy$. Therefore, we obtain 
\begin{gather*}
    W_{(q)}(x,y)
    =
    e^{\frac{(x^2+y^2)\theta}{2}}
    \left[
    D_{-\frac{q}{\theta}}(-\sqrt{2\theta}x)
    D_{-\frac{q}{\theta}}(\sqrt{2\theta}y)
    -
    D_{-\frac{q}{\theta}}(-\sqrt{2\theta}y)
    D_{-\frac{q}{\theta}}(\sqrt{2\theta}x)
    \right].
\end{gather*}
\subsection{Diffusions with the exponentially decayed mean-reversion drift}
\cite{dassios2018economic} proposed a type of stochastic process, which is the logarithm of the Shiryaev process. It is governed by the SDE, defined as 
\begin{gather*}
    X_t=x+\int_0^t\nu(e^{-2lX_t}-1)ds+B_t,\quad x\in \mathbb{R},
\end{gather*}
where $\nu> 0 $ and $l>0$. Two corresponding eigenfunctions are given by 
\begin{gather*}
\phi_q^+(x)=e^{x(\nu-\sqrt{\nu^2+2q})}M(\frac{\sqrt{\nu^2+2q}-\nu}{2l},\frac{\sqrt{\nu^2+2q}+l}{l},\frac{\nu e^{-2lx}}{l}),\quad x\in \mathbb{R},
\\
\phi_q^-(x)=e^{x(\nu-\sqrt{\nu^2+2q})}U(\frac{\sqrt{\nu^2+2q}-\nu}{2l},\frac{\sqrt{\nu^2+2q}+l}{l},\frac{\nu e^{-2lx}}{l}),\quad x\in \mathbb{R},
\end{gather*}
where $M$ and $U$ are the confluent hypergeometric function of the first kind and the second kind.
\bibliographystyle{alpha}
\bibliography{references}

\begin{thebibliography}{APWY16}

\bibitem[AKP04]{avram2004exit}
Florin Avram, Andreas~E Kyprianou, and Martijn~R Pistorius.
\newblock Exit problems for spectrally negative {L{\'e}vy} processes and
  applications to (canadized) russian options.
\newblock {\em The Annals of Applied Probability}, 14(1):215--238, 2004.

\bibitem[APP07]{avram2007optimal}
Florin Avram, Zbigniew Palmowski, and Martijn~R Pistorius.
\newblock On the optimal dividend problem for a spectrally negative {L{\'e}vy}
  process.
\newblock {\em The Annals of Applied Probability}, 17(1):156--180, 2007.

\bibitem[APWY16]{avanzi2016optimal}
Benjamin Avanzi, Jos{\'e}-Luis P{\'e}rez, Bernard Wong, and Kazutoshi Yamazaki.
\newblock On optimal joint reflective and refractive dividend strategies in
  spectrally positive {L{\'e}vy} processes.
\newblock {\em UNSW Business School Research Paper}, (2016ACTL05), 2016.

\bibitem[BE10]{bayraktar2010one}
Erhan Bayraktar and Masahiko Egami.
\newblock On the one-dimensional optimal switching problem.
\newblock {\em Mathematics of Operations Research}, 35(1):140--159, 2010.

\bibitem[BK10]{BIFFIS201085}
Enrico Biffis and Andreas~E. Kyprianou.
\newblock A note on scale functions and the time value of ruin for {L{\'e}vy}
  insurance risk processes.
\newblock {\em Insurance: Mathematics and Economics}, 46(1):85--91, 2010.

\bibitem[BKY14]{bayraktar2014optimal}
Erhan Bayraktar, Andreas~E Kyprianou, and Kazutoshi Yamazaki.
\newblock Optimal dividends in the dual model under transaction costs.
\newblock {\em Insurance: Mathematics and Economics}, 54:133--143, 2014.

\bibitem[BS15]{borodin2015handbook}
Andrei~N Borodin and Paavo Salminen.
\newblock {\em Handbook of Brownian motion-facts and formulae}.
\newblock Springer Science \& Business Media, 2015.

\bibitem[CP14]{czarna2014dividend}
Irmina Czarna and Zbigniew Palmowski.
\newblock Dividend problem with parisian delay for a spectrally negative
  {L{\'e}vy} risk process.
\newblock {\em Journal of Optimization Theory and Applications},
  161(1):239--256, 2014.

\bibitem[DL18]{dassios2018economic}
Angelos Dassios and Luting Li.
\newblock An economic bubble model and its first passage time.
\newblock {\em arXiv preprint arXiv:1803.08160}, 2018.

\bibitem[Fer19]{ferrari2019class}
Giorgio Ferrari.
\newblock On a class of singular stochastic control problems for reflected
  diffusions.
\newblock {\em Journal of Mathematical Analysis and Applications},
  473(2):952--979, 2019.

\bibitem[KI99]{ken1999levy}
Sato Ken-Iti.
\newblock {\em L{\'e}vy processes and infinitely divisible distributions}.
\newblock Cambridge university press, 1999.

\bibitem[Leh77]{lehoczky1977formulas}
John~P Lehoczky.
\newblock Formulas for stopped diffusion processes with stopping times based on
  the maximum.
\newblock {\em The Annals of Probability}, pages 601--607, 1977.

\bibitem[LR10]{loeffen2010finetti}
Ronnie~L Loeffen and Jean-Fran{\c{c}}ois Renaud.
\newblock De finetti’s optimal dividends problem with an affine penalty
  function at ruin.
\newblock {\em Insurance: Mathematics and Economics}, 46(1):98--108, 2010.

\bibitem[Nob21]{NOBA202173}
Kei Noba.
\newblock On the optimality of double barrier strategies for {L\'evy}
  processes.
\newblock {\em Stochastic Processes and their Applications}, 131:73--102, 2021.

\bibitem[NPYY18]{noba2018optimal}
Kei Noba, Jos{\'e}-Luis P{\'e}rez, Kazutoshi Yamazaki, and Kouji Yano.
\newblock On optimal periodic dividend strategies for {L{\'e}vy} risk
  processes.
\newblock {\em Insurance: Mathematics and Economics}, 80:29--44, 2018.

\bibitem[YW12]{yin2012first}
Chuancun Yin and Huiqing Wang.
\newblock The first passage time and the dividend value function for
  one-dimensional diffusion processes between two reflecting barriers.
\newblock {\em International Journal of Stochastic Analysis}, 2012.

\bibitem[Zha15]{zhang2015occupation}
Hongzhong Zhang.
\newblock Occupation times, drawdowns, and drawups for one-dimensional regular
  diffusions.
\newblock {\em Advances in Applied Probability}, 47(1):210--230, 2015.

\bibitem[ZY16]{zhu2016optimal}
Jinxia Zhu and Hailiang Yang.
\newblock Optimal capital injection and dividend distribution for growth
  restricted diffusion models with bankruptcy.
\newblock {\em Insurance: Mathematics and Economics}, 70:259--271, 2016.

\end{thebibliography}
\end{document}